\documentclass[11pt,a4paper,reqno]{amsart}%
\usepackage[latin1]{inputenc}
\usepackage{mathrsfs}
\usepackage{dsfont}
\usepackage{hyperref}
\usepackage{amsmath}
\usepackage{amssymb}
\usepackage{amsthm}
\usepackage{amsfonts}
\usepackage{amstext}
\usepackage{amsopn}
\usepackage{amsxtra}
\usepackage{mathrsfs}
\usepackage{dsfont}
\usepackage{esint}
\usepackage{pst-all}
\usepackage{pstricks}
\usepackage{graphicx}

\theoremstyle{plain}
\newtheorem{theorem}{Theorem}

\theoremstyle{definition}

\theoremstyle{remark}
\newtheorem{remark}{Remark}

\numberwithin{equation}{section}

\newcommand{\ii}{\infty}
\newcommand\R{{\ensuremath {\mathbb R} }}
\newcommand\C{{\ensuremath {\mathbb C} }}

\newcommand\bP{{\ensuremath {\mathds P} }}
\newcommand\bQ{{\ensuremath {\mathds Q} }}
\renewcommand\phi{\varphi}

\newcommand{\gS}{\mathfrak{S}}

\newcommand{\wto}{\rightharpoonup}

\newcommand{\cE}{\mathcal{E}}

\newcommand{\eps}{\epsilon}

\renewcommand{\epsilon}{\varepsilon}
\newcommand\pscal[1]{{\ensuremath{\left\langle #1 \right\rangle}}}
\newcommand{\norm}[1]{ \left| \! \left| #1 \right| \! \right| }
\newcommand{\tr}{{\rm Tr}\,}

\renewcommand{\geq}{\geqslant}
\renewcommand{\leq}{\leqslant}

\newcommand{\nn}{\nonumber}

\title[Semi-classical limit of the Levy-Lieb functional]{Semi-classical limit of the Levy-Lieb functional in Density Functional Theory}

\author[M. Lewin]{Mathieu Lewin}
\address{CNRS \& CEREMADE, Universit\'e Paris-Dauphine, PSL Research University, F-75016 Paris, France} 
\email{mathieu.lewin@math.cnrs.fr}

\date{Final version to appear in \emph{Comptes rendus de l'Acad\'emie des Sciences, Math\'ematiques}\ \copyright~2017 by the author. This paper may be reproduced, in its entirety, for non-commercial purposes}

\begin{document}

\maketitle

\begin{abstract}
In a recent work, Bindini and De Pascale have introduced a regularization of $N$-particle symmetric probabilities which preserves their one-particle marginals. In this short note, we extend their construction to mixed quantum fermionic states. This enables us to prove the convergence of the Levy-Lieb functional in Density Functional Theory, to the corresponding multi-marginal optimal transport in the semi-classical limit. Our result holds for mixed states of any particle number $N$, with or without spin.
\end{abstract}

\section{Extending the Bindini-De Pascale construction}

Let $\bP$ be a symmetric  $N$-particle probability measure over $(\R^d)^N$ and let 
$$\rho_\bP(x_1)=\int_{(\R^d)^{N-1}}d\bP(x_1,x_2,...,x_N)$$
be its one-particle marginal. We assume that $\sqrt{\rho_\bP}\in H^1(\R^d)$, as is appropriate in Density Functional Theory (DFT)~\cite{Harriman-81,Lieb-83b}. An interesting question, important for applications in DFT, is to approximate $\bP$ by a regular $N$-particle probability density $\bP_\eps$ with the same density $\rho_{\bP_\eps}=\rho_\bP$. A more challenging problem, considered first in~\cite{CotFriKlu-13,BinPas-17} for $N=2,3$ and solved for all $N\geq2$ in this note, is to find a fermionic quantum state $\Gamma_\eps$ with the same density $\rho_{\Gamma_\eps}=\rho_\bP$ and a controlled kinetic energy.

Consider a radial function $\chi\in C^\ii_c(\R^d,\R)$ with support in the unit ball of $\R^d$, such that $\int_{\R^d}\chi^2=1$, and denote $\chi_\eps(x)=\eps^{-d/2}\chi(x/\eps)$. In~\cite{BinPas-17}, Bindini and De Pascale have introduced the following elegant regularization 
\begin{multline}
\bP_\eps(x_1,...,x_N)\\=\iint_{\R^{2dN}}\prod_{k=1}^N\frac{\rho_\bP(x_k)\chi_\eps(x_k-z_k)^2\chi_\eps(z_k-y_k)^2}{\rho_\bP\ast\chi_\eps^2(z_k)}d\bP(y_1,...,y_N)\;dz_1\cdots dz_N.
\label{eq:Bindini-Pascale}
\end{multline}
We assume in the following that $\bP$ has its support in
$$D_\alpha^c:=\left\{X=(x_1,...,x_N)\in (\R^d)^N\ :\ |x_i-x_j|\geq \alpha,\ \forall i\neq j\right\},$$
for some $\alpha>0$, a condition which is satisfied for minimizers of the Coulomb $N$-particle energy~\cite{ButChaPas-17}. Since $\chi$ is supported in the unit ball, $\bP_\eps$ is then supported on the set $D^c_{\alpha-4\eps}$ where all the particles are at a distance $\alpha-4\eps$ to each other. In the following we always assume that $\eps<\alpha/4$.
The purpose of the integration over the $y_k$'s in~\eqref{eq:Bindini-Pascale} is to regularize the probability $\bP$, since the convolution
$$\bQ_\eps(z_1,...,z_N)=\int_{(\R^d)^N}\prod_{k=1}^N\chi_\eps(z_k-y_k)^2d\bP(y_1,...,y_N)$$
is now $C^\ii$. However its density is $\rho_{\bQ_\eps}=\rho_\bP\ast\chi_\eps^2$
and the purpose of the integration over the $z_k$'s is to map back the density to $\rho_\bP$. Indeed, integrating~\eqref{eq:Bindini-Pascale} over $x_2,...,x_N$, we get $\prod_{k=2}^N\rho_\bP\ast\chi^2_\eps(z_k)$ in the numerator, which cancels with the denominator. The corresponding integrals over $z_2,...,z_N$ give $(\int\chi_\eps^2)^{N-1}=1$ and we end up with
\begin{align*}
\rho_{\bP_\eps}(x_1) =\rho_{\bP}(x_1)\int_{\R^d}\int_{\R^d}\frac{\chi_\eps(x_1-z_1)^2\chi_\eps(z_1-y_1)^2}{\rho_{\bP}\ast\chi_\eps^2(z_1)}\rho_{\bP}(y_1)\,dy_1\,dz_1=\rho_\bP(x_1).
\end{align*}
A somewhat different method was introduced in~\cite{CotFriKlu-13} by Cotar, Friesecke and Kluppelberg. 

In~\cite{BinPas-17}, Bindini and De Pascale prove that 
\begin{equation}
 \int_{\R^{dN}}|\nabla\sqrt{\bP_\eps}|^2\leq N\left(\int_{\R^d}|\nabla\sqrt{\rho_\bP}|^2+\frac{1}{\eps^2}\int_{\R^d}|\nabla\chi|^2\right)
 \label{eq:estim_kinetic_Bindini_dePascale}
\end{equation}
and use this to get some information on the semi-classical limit of the Levy-Lieb functional (to be discussed later in Section~\ref{sec:Levy-Lieb}). 
Unfortunately, the probability~\eqref{eq:Bindini-Pascale} is really a classical object. Although for bosons one can use the symmetric wavefunction $\Psi_\eps=\sqrt{\bP_\eps}$, for fermions the wavefunction must be anti-symmetric with respect to the permutations of its $N$ variables. In space dimensions $d=1$ and $d=2$, one can use the multiplication by 
$$U(x_1,...,x_N)=\prod_{1\leq j<k\leq N}\frac{x_j-x_k}{|x_j-x_k|}$$
which maps bosons onto fermions and conversely, whatever the value of $N$ (if $d=2$ we identify $\R^2$ with $\C$). Since $U$ is $C^\ii$ on $D_\alpha^c$ where $\bP_\eps$ is supported, $U\sqrt{\bP_\eps}$ has the same regularity as $\sqrt{\bP_\eps}$ and satisfies an estimate similar to~\eqref{eq:estim_kinetic_Bindini_dePascale}, with a worse dependence in $N$. In dimension $d\geq3$, the situation is more complex, due to some well known topological obstructions~\cite{LeiMyr-77,Myrheim-99,Ouvry-07}. Indeed, for $N\geq2$ and $d\geq3$, there does not exist any anti-symmetric function $U:D_\alpha^c\subset (\R^d)^N\to \C$ which is continuous and satisfies $|U(x_1,...,x_N)|=1$. Otherwise, consider for instance the odd function
$$z\mapsto U(y+z,y-z,x_3,...,x_N)$$
for fixed $y,x_2,...,x_N$. By the Borsuk-Ulam theorem, it must vanish on any sphere $\{|z|=R\}$, and this would contradict $|U|=1$ for $|y|$ and $R$ large enough. In~\cite{CotFriKlu-13,BinPas-17} the authors use the spin variable to antisymmetrize $\sqrt{\bP_\eps}$ but, in dimension $d=3$, this has so far limited the results to $N=2$ and $N=3$.

Our idea in this short note is to overcome these difficulties using the concept of mixed states. We propose the following simple quantum extension of~\eqref{eq:Bindini-Pascale}, for fermions:
\begin{multline}
\Gamma_\eps
=\iint_{\R^{dN}\times \R^{dN}}
\sqrt{\rho_\bP}^{\otimes N}\;|\chi_{\eps,z_1}\wedge\cdots\wedge\chi_{\eps,z_N}\rangle\langle\chi_{\eps,z_1}\wedge\cdots\wedge\chi_{\eps,z_N}|\;\sqrt{\rho_\bP}^{\otimes N}\\
\times\prod_{k=1}^N\frac{\chi_\eps(z_k-y_k)^2}{\rho_\bP\ast\chi_\eps^2(z_k)}d\bP(y_1,...,y_N)\;dz_1\cdots dz_N.
\label{eq:quantum_state}
\end{multline}
Here $\Gamma_\eps$ is a non-negative self-adjoint operator which acts on the $N$-particle fermionic space $L^2_a((\R^d)^N)=\bigwedge_1^N L^2(\R^d)$. For simplicity we do not consider spin here. It is easy to extend the trial state~\eqref{eq:quantum_state} to the case of particles with $q$ spin states, for instance by putting all the particles in the same spin state. One can also construct any desired eigenvector of the total spin operator by adding an appropriate spin state to each function $\chi_{\eps,z_k}$.

In~\eqref{eq:quantum_state}, we use the notation $\chi_{\eps,z}(x)=\chi_\eps(x-z)$ and recall that $\eps<\alpha/4$ such that the functions $\chi_{\eps,z_1},...,\chi_{\eps,z_N}$ have disjoint supports.  We use the ket-bra notation $|\Psi\rangle\langle\Psi|$ for the orthogonal projection on $\Psi$, defined by $|\Psi\rangle\langle\Psi|\phi=\pscal{\Psi,\phi}\Psi$. We call
$$
\phi_1\wedge\cdots\wedge\phi_N(x_1,...,x_N)=\frac{1}{\sqrt{N!}}\sum_{\sigma\in\gS_N}\eps(\sigma)\phi_{{\sigma(1)}}(x_1)\cdots \phi_{\sigma(N)}(x_N)
$$
the Slater determinant. Finally, the $N$-fold tensor product is defined by
$$\sqrt{\rho_\bP}^{\otimes N}(x_1,...,x_N)=\prod_{j=1}^N\sqrt{\rho_\bP}(x_j).$$
In~\eqref{eq:quantum_state}, $\sqrt{\rho_\bP}^{\otimes N}$ is understood as a multiplication operator on $L^2_a((\R^d)^N)$. In particular, we have
\begin{multline*}
\sqrt{\rho_\bP}^{\otimes N}\;|\chi_{\eps,z_1}\wedge\cdots\wedge\chi_{\eps,z_N}\rangle\langle\chi_{\eps,z_1}\wedge\cdots\wedge\chi_{\eps,z_N}|\;\sqrt{\rho_\bP}^{\otimes N}\\
=|\sqrt{\rho_\bP}\chi_{\eps,z_1}\wedge\cdots\wedge\sqrt{\rho_\bP}\chi_{\eps,z_N}\rangle\langle\sqrt{\rho_\bP}\chi_{\eps,z_1}\wedge\cdots\wedge\sqrt{\rho_\bP}\chi_{\eps,z_N}|.
\end{multline*}
The integral kernel of $\Gamma_\eps$ is given by 
\begin{multline}
\Gamma_\eps(x_1,...,x_N;x_1',...,x_N')=\frac{1}{N!}\iint_{\R^{dN}\times \R^{dN}}\det\big(\chi_{\eps,z_i}(x_j)\big)\det\big(\chi_{\eps,z_i}(x'_j)\big)\\
\times\prod_{k=1}^N\frac{\sqrt{\rho_\bP}(x_k)\sqrt{\rho_\bP}(x'_k)\chi(z_k-y_k)^2}{\rho_\bP\ast\chi_\eps^2(z_k)}d\bP(y_1,...,y_N)\;dz_1\cdots dz_N.
\label{eq:quantum_state_kernel}
\end{multline}
Since the $\chi_{\eps,z_k}$ have disjoint supports, we have
\begin{align}
\Big(\det\big(\chi_{\eps,z_i}(x_j)\big)\Big)^2&= \sum_{\sigma,\sigma'\in\gS_N}\eps(\sigma)\,\eps(\sigma')\prod_{k=1}^N \chi_{\eps,z_{\sigma(k)}}(x_k)\chi_{\eps,z_{\sigma'(k)}}(x_k)\nn\\
&=\sum_{\sigma\in\gS_N}\prod_{k=1}^N \chi_{\eps}(x_k-z_{\sigma(k)})^2.
\label{eq:formule_carre}
\end{align}
Using the symmetry of $\bP$, we have
$$\int_{\R^{dN}}\prod_{k=1}^N\chi_\eps(z_k-y_k)^2d\bP(y_1,...,y_N)=\int_{\R^{dN}}\prod_{k=1}^N\chi_\eps(z_{\sigma(k)}-y_k)^2d\bP(y_1,...,y_N)$$
for every $\sigma\in\gS_N$. Therefore, in~\eqref{eq:formule_carre} the terms in the sum over the permutations $\sigma\in\gS_N$ all contribute the same 	amount. 
We thus find that the diagonal of~$\Gamma_\eps$ coincides with the Bindini-De Pascale probability density:
$$\Gamma_\eps(x_1,...,x_N;x_1,...,x_N)=\bP_\eps(x_1,...,x_N).$$
From this we conclude that
\begin{multline*}
\tr(\Gamma_\eps)=\int_{(\R^d)^N}\Gamma_\eps(x_1,...,x_N;x_1,...,x_N)\,dx_1\cdots dx_N\\
=\int_{\R^{dN}}d\bP_\eps(x_1,...,x_N)=1. 
\end{multline*}
and $\Gamma_\eps$ is a proper fermionic (mixed) state. We recall that the one-particle density of $\Gamma_\eps$ is defined by duality, requiring that $\tr(\phi(x_1)\Gamma_\eps)=\int_{\R^d}\phi\,\rho_{\Gamma_\eps}$ for every $\phi\in L^\ii(\R^d)$. For a continuous kernel such as $\Gamma_\eps$, we have
$$\rho_{\Gamma_\eps}(x_1)=\int_{\R^{dN}}\Gamma_\eps(x_1,...,x_N;x_1,...,x_N)\,dx_2\cdots dx_N=\rho_{\bP_\eps}(x_1)=\rho_{\bP}(x_1).$$
From the Cauchy-Schwarz inequality we have
$$\int_{(\R^d)^N}|\nabla\sqrt{\bP_\eps}|^2\leq \tr(-\Delta)\Gamma_\eps,$$
which is in the wrong direction to conclude anything about the kinetic energy of $\Gamma_\eps$ using the estimate~\eqref{eq:estim_kinetic_Bindini_dePascale} of Bindini-De Pascale. But we can prove the following theorem, which implies~\eqref{eq:estim_kinetic_Bindini_dePascale}.

\begin{theorem}[Estimates on $\Gamma_\eps$]\label{thm:estimates}
Let $\bP$ be a symmetric $N$-particle density with support in $D^c_\alpha$ for some $\alpha>0$ and such that $\sqrt{\rho_\bP}\in H^1(\R^d)$. Let $\Gamma_\eps$ be defined by~\eqref{eq:quantum_state}. Then, for $\eps<\alpha/4$ we have
\begin{equation}
 \tr(-\Delta)\Gamma_\eps = N\left(\int_{\R^d}|\nabla\sqrt{\rho_\bP}|^2+\frac{1}{\eps^2}\int_{\R^d}|\nabla\chi|^2\right).
 \label{eq:estim_kinetic}
\end{equation}
In addition, for every symmetric function $\Phi\in C^2(D^c_{\alpha-4\eps})$, we have
\begin{multline}
 \left|\tr(\Phi\Gamma_\eps)-\int_{\R^{dN}}\Phi\, d\bP\right|=\left|\int_{\R^{dN}}\Phi\, d\bP_\eps-\int_{\R^{dN}}\Phi\, d\bP\right|\\
 \leq \eps^2\,\Bigg\{\sum_{j=1}^N\norm{\nabla_j\Phi}_{L^\ii(D^c_{\alpha-4\eps})}\left(\int_{\R^{d}}|\nabla\rho_\bP|\right)\left(\int_{\R^d}|u|^2\chi(u)^2\,du\right)\\
 +2\sum_{j,k=1}^N\norm{\nabla_{j}\nabla_{k}\Phi}_{L^\ii(D^c_{\alpha-4\eps})}\Bigg\}.
 \label{eq:estim_potential}
\end{multline}
In particular, $\bP_\eps\wto\bP$, as was proved in~\cite{BinPas-17}.
\end{theorem}

\begin{proof}
We have
\begin{multline*}
\tr(-\Delta)\Gamma_\eps= \sum_{\ell=1}^N\iint_{\R^{dN}\times \R^{dN}}
\left(\int_{\R^{dN}}\left|\nabla_{x_\ell}\left( \sqrt{\rho_\bP}^{\otimes N}\chi_{\eps,z_1}\wedge\cdots\wedge\chi_{\eps,z_N}\right)\right|^2\right)\times\\
\times\prod_{k=1}^N\frac{\chi_\eps(z_k-y_k)^2}{\rho_\bP\ast\chi_\eps^2(z_k)}d\bP(y_1,...,y_N)\;dz_1\cdots dz_N.
\end{multline*}
Similarly as in~\eqref{eq:formule_carre}, 
\begin{multline*}
\left|\nabla_{x_\ell}\left(\sqrt{\rho_\bP}\chi_{\eps,z_1}\wedge\cdots\wedge\sqrt{\rho_\bP}\chi_{\eps,z_N}\right)\right|^2\\
=\frac{1}{N!}\sum_{\sigma\in\gS_N}\left|\nabla\big(\sqrt{\rho_\bP}\chi_{\eps,z_{\sigma(\ell)}}\big)(x_\ell)\right|^2\prod_{\substack{k=1,...,N\\ k\neq\ell}} \rho_\bP(x_k)\chi_{\eps,z_{\sigma(k)}}(x_k)^2. 
\end{multline*}
Integrating over $x_1,...,x_N$, we obtain
\begin{multline*}
\sum_{\ell=1}^N\int_{\R^{dN}}\left|\nabla_{x_\ell}\left( \sqrt{\rho_\bP}^{\otimes N}\chi_{\eps,z_1}\wedge\cdots\wedge\chi_{\eps,z_N}\right)\right|^2\\
=\sum_{j=1}^N\left(\int_{\R^d}\left|\nabla\big(\sqrt{\rho_\bP}\chi_{\eps,z_{j}}\big)(x)\right|^2\,dx\right)\prod_{\substack{k=1,...,N\\ k\neq j}}\rho_\bP\ast \chi_\eps^2(z_j).
\end{multline*}
Finally, integrating over the $y_k$'s and $z_k$'s, we conclude that
\begin{align*}
\tr(-\Delta)\Gamma_\eps&=N\int_{\R^d}\int_{\R^d}\left|\nabla\big(\sqrt{\rho_\bP}\chi_{\eps,z}\big)(x)\right|^2\,dx\,dz\\
&=N\left(\int_{\R^d}|\nabla\sqrt{\rho_\bP}(x)|^2\,dx+\frac{1}{\eps^2}\int_{\R^d}|\nabla\chi(x)|^2\,dx\right). 
\end{align*}
To prove~\eqref{eq:estim_potential}, we remark that 
$$\iint_{\R^{2dN}}\prod_{k=1}^N\frac{\rho_\bP(x_k)\chi_\eps(x_k-z_k)^2\chi_\eps(z_k-y_k)^2}{\rho_\bP\ast\chi_\eps^2(z_k)}dx_1\cdots dx_N\;dz_1\cdots dz_N=1.$$
Hence, using the shorter notation $X=(x_1,...,x_N)$, $Y=(y_1,...,y_N)$ and $Z=(z_1,...,z_N)$, we have
\begin{multline*}
\int_{\R^{dN}}\Phi\, d\bP_\eps-\int_{\R^{dN}}\Phi\, d\bP\\
=\iint_{\R^{3dN}}\big(\Phi(X)-\Phi(Y)\big)\prod_{k=1}^N\frac{\rho_\bP(x_k)\chi_\eps(x_k-z_k)^2\chi_\eps(z_k-y_k)^2}{\rho_\bP\ast\chi_\eps^2(z_k)}
d\bP(Y)\;dX\;dZ.
\end{multline*}
Since $|x_j-y_j|\leq 2\eps$ due to the support of $\chi$, we conclude from the fundamental theorem of calculus that
\begin{multline*}
\bigg|\int_{\R^{dN}}\Phi\, d\bP_\eps-\int_{\R^{dN}}\Phi\, d\bP
-\iint_{\R^{3dN}}d\bP(Y)\,dX\,dZ\left(\sum_{\ell=1}^N\nabla_{\ell}\Phi(Y)\cdot (x_\ell-y_\ell)\right)\times\\
\times\prod_{k=1}^N\frac{\rho_\bP(x_k)\chi_\eps(x_k-z_k)^2\chi_\eps(z_k-y_k)^2}{\rho_\bP\ast\chi_\eps^2(z_k)}
\bigg|
\leq 2\eps^2\sum_{\ell,m=1}^N\norm{\nabla_{x_\ell}\nabla_{x_m}\Phi}_{L^\ii(D^c_{\alpha-4\eps})}.
\end{multline*}
It remains to estimate the term involving $\nabla_\ell\Phi$. By symmetry, it is sufficient to look at the case $\ell=1$, which can be expressed in the form
\begin{multline*}
\iint_{\R^{3dN}}\nabla_{1}\Phi(Y)\cdot (x_1-y_1)\prod_{k=1}^N\frac{\rho_\bP(x_k)\chi_\eps(x_k-z_k)^2\chi_\eps(z_k-y_k)^2}{\rho_\bP\ast\chi_\eps^2(z_k)}d\bP(Y)\,dX\,dZ\\
=\iint_{\R^{dN+2d}}\nabla_{1}\Phi(Y)\cdot (x_1-y_1)\frac{\rho_\bP(x_1)\chi_\eps(x_1-z_1)^2\chi_\eps(z_1-y_1)^2}{\rho_\bP\ast\chi_\eps^2(z_1)}d\bP(Y)\,dx_1\,dz_1.
\end{multline*}
Here we can replace $x_1-y_1$ by $x_1-z_1$ since 
$\int_{\R^d}(z_1-y_1)\chi_\eps(z_1-y_1)^2\,dz_1=0,$ 
due to the symmetry of $\chi$. We then estimate 
\begin{align*}
&\bigg|\iint_{\R^{3dN}}\nabla_{1}\Phi(Y)\cdot (x_1-z_1)\frac{\rho_\bP(x_1)\chi_\eps(x_1-z_1)^2\chi_\eps(z_1-y_1)^2}{\rho_\bP\ast\chi_\eps^2(z_1)}d\bP(Y)\,dx_1\,dz_1\bigg|\\[0.2cm]
&\ \leq\norm{\nabla_1\Phi}_{L^\ii(D^c_{\alpha-4\eps})}\times\\
&\quad\times\iint_{\R^{2d}}\left|\int_{\R^d}(x_1-z_1)\rho_\bP(x_1)\chi_\eps(x_1-z_1)^2\,dx_1\right|\frac{\chi_\eps(z_1-y_1)^2}{\rho_\bP\ast\chi_\eps^2(z_1)}\,\rho_{\bP}(y_1)\,dz_1\,dy_1\\
&\ = \norm{\nabla_1\Phi}_{L^\ii(D^c_{\alpha-4\eps})}\int_{\R^{d}}\left|\int_{\R^d}(x-z)\rho_\bP(x)\chi_\eps(x-z)^2\,dx\right|\,dz.
\end{align*}
By the symmetry of $\chi$ and the fundamental theorem of calculus, we have 
\begin{align*}
&\left|\int_{\R^d}(x-z)\rho_\bP(x)\chi_\eps(x-z)^2\,dx\right|\\
&\qquad\qquad=\left|\int_{\R^d}(x-z)\big(\rho_\bP(x)-\rho_\bP(z)\big)\chi_\eps(x-z)^2\,dx\right|\\
&\qquad\qquad\leq \int_0^1\int_{\R^d}|\nabla\rho_\bP(z+tu)|\,|u|^2\chi_\eps(u)^2\,du.
\end{align*}
Integrating over $z$ we find the claimed estimate
\begin{multline*}
\bigg|\iint_{\R^{3dN}}\nabla_{1}\Phi(Y)\cdot (x_1-y_1)\prod_{k=1}^N\frac{\rho_\bP(x_k)\chi_\eps(x_k-z_k)^2\chi_\eps(z_k-y_k)^2}{\rho_\bP\ast\chi_\eps^2(z_k)}d\bP(Y)\,dX\,dZ\bigg|\\
\leq \eps^2\norm{\nabla_1\Phi}_{L^\ii(D^c_{\alpha-4\eps})}\left(\int_{\R^{d}}|\nabla\rho_\bP|\right)\left(\int_{\R^d}|u|^2\chi(u)^2\,du\right).
\end{multline*}
\end{proof}

\section{Semiclassical limit of the Levy-Lieb functional}\label{sec:Levy-Lieb}

Here we restrict ourselves for simplicity to the physical space $\R^3$ and the Coulomb potential. 
Density Functional Theory is based on the following functional~\cite{Levy-79,Lieb-83b,CanDefKutLeBMad-03}
\begin{equation}
\cE(\rho)=\min_{\substack{\Gamma=\Gamma^*\geq0\\ \tr(\Gamma)=1\\ \rho_\Gamma=\rho}}\tr\left(-\sum_{j=1}^N\Delta_{x_j}+\sum_{1\leq j<k\leq N}\frac{1}{|x_j-x_k|}\right)\Gamma,
\label{eq:Levy-Lieb}
\end{equation}
of the density $\rho$, a given non-negative function such that $\int_{\R^3}\rho=N$ and $\sqrt{\rho}\in H^1(\R^3)$. In the minimum $\Gamma$ is an operator acting on the fermionic space $\bigwedge_1^NL^2(\R^3)$. Motivated by arguments of Hohenberg and Kohn~\cite{HohKoh-64}, Levy introduced in~\cite{Levy-79} a functional similar to~\eqref{eq:Levy-Lieb} but with the additional constraint that $\Gamma=|\Psi\rangle\langle\Psi|$ is a rank-one orthogonal projection (pure state). The latter was rigorously studied by Lieb in~\cite{Lieb-83b}, who proposed to extend the definition to mixed states, as in~\eqref{eq:Levy-Lieb}. The minimum over mixed states~\eqref{eq:Levy-Lieb} has better mathematical properties than with pure states~\cite{Lieb-83b}. For instance, $\cE$ is convex and, by the linearity in $\Gamma$, we have the dual formulation
\begin{multline*}
\cE(\rho)=\sup \bigg\{\int_{\R^3}\rho(x)\,V(x)\,dx\ :\ V\in L^{3/2}(\R^3,\R)+L^\ii(\R^3,\R),\\ -\sum_{j=1}^N\Delta_{x_j}-\sum_{j=1}^NV(x_j)+\sum_{1\leq j<k\leq N}\frac{1}{|x_j-x_k|}\geq0\bigg\},
\end{multline*}
which is the quantum equivalent of the Kantorovich duality used in optimal transport~\cite{Villani-09}. The last inequality is in the sense of self-adjoint operators. 

It is possible to introduce an effective semi-classical parameter $\eta=\hbar^2$ by scaling the density $\rho$. Namely, for $\rho_\eta(x)=\eta^{3}\rho(\eta x)$ we have 
\begin{equation}
\frac{\cE(\rho_\eta)}{\eta}=\min_{\substack{\Gamma=\Gamma^*\geq0\\ \tr(\Gamma)=1\\ \rho_\Gamma=\rho}}\tr\left(-\eta\sum_{j=1}^N\Delta_{x_j}+\sum_{1\leq j<k\leq N}\frac{1}{|x_j-x_k|}\right)\Gamma.
\label{eq:Levy-Lieb_eta}
\end{equation}
In the limit $\eta\to0$, we prove the convergence to the Coulomb multi-marginal optimal transport problem
\begin{equation}
\cE_{\rm OT}(\rho)=\min_{\substack{\bP\ \text{symmetric}\\ \text{probability on $(\R^3)^N$}\\ \rho_\bP=\rho}}\int_{(\R^3)^N}\sum_{1\leq j<k\leq N}\frac{1}{|x_j-x_k|}d\bP(x_1,...,x_N),
\label{eq:OT}
\end{equation}
which has recently received a lot of attention~\cite{ButPasGor-12,CotFriKlu-13,CotFriPas-15,ColMar-15,MarGerNen-15,SeiMarGerNenGieGor-17}

\begin{theorem}[Semi-classical limit]
Let $N\geq2$ and let $\rho\geq0$ be such that $\int_{\R^3}\rho=N$ and $\sqrt{\rho}\in H^1(\R^3)$. Then we have for a constant $C$ (depending on $N$ and $\rho$)
\begin{equation}
\cE_{\rm OT}(\rho)\leq \frac{\cE\left(\eta^3\rho(\eta\,\cdot)\right)}\eta\leq \cE_{\rm OT}(\rho)+C(\sqrt{\eta}+\eta).
\label{eq:estim_semiclassics}
\end{equation}
In particular,
$$\lim_{\eta\to0}\frac{\cE\left(\eta^3\rho(\eta\,\cdot)\right)}\eta= \cE_{\rm OT}(\rho).$$
\end{theorem}

This theorem generalizes the results in~\cite{CotFriKlu-13,BinPas-17} for $N=2,3$ in the pure state case. Results for $N\geq3$ have been announced 
in~\cite[Ref. 7]{FriMenPasCotKlu-13} but they were not yet available at the time this note was written. It would be interesting to extend our findings to pure states.

Semi-classical analysis suggests that the behavior in $\sqrt\eta$ is optimal for small $\eta$. The next order (in $\sqrt{\eta}=\hbar$) in the expansion of $\cE(\rho_\eta)/\eta$ was predicted in~\cite{GorVigSei-09}.

\begin{proof}
Let $\bP$ be an optimizer for $\cE_{\rm OT}(\rho)$. It has been shown in~\cite{ButChaPas-17} that $\bP$ has its support on $D^c_\alpha$ for some $\alpha>0$. Taking then our $\Gamma_\eps$ as a trial state, we find by Theorem~\ref{thm:estimates}
\begin{multline*}
\frac{\cE(\rho_\eta)}\eta \leq \eta N\left(\int_{\R^d}|\nabla\sqrt{\rho}|^2+\frac{1}{\eps^2}\int_{\R^d}|\nabla\chi|^2\right)+\cE_{\rm OT}(\rho)\\
+C\eps^2\left(\frac{N^3}{(\alpha-4\eps)^2}\int_{\R^{d}}|\nabla\rho_\bP|\int_{\R^d}|u|^2\chi(u)^2\,du+\frac{N^4}{(\alpha-4\eps)^3}\right). 
\end{multline*}
We have used here that $\Phi(X)=\sum_{1\leq j<k\leq N}|x_j-x_k|^{-1}$ is $C^\ii$ on $D^c_{\alpha-4\eps}$.
Optimizing in $\eps$ gives the result.
\end{proof}

\begin{remark}
The convergence of states in the limit $\eta\to0$ can be proved as in~\cite{BinPas-17}.
\end{remark}

\subsection*{Acknowledgement} This project has received funding from the European Research Council (ERC) under the European Union's Horizon 2020 research and innovation programme (grant agreement MDFT No 725528).


\end{document}